\documentclass[11pt, a4paper]{article}%
\usepackage[margin=2.35cm]{geometry}
\usepackage{authblk}
\usepackage{graphicx}
\usepackage{subcaption}

\usepackage{url}
\usepackage{amssymb,mathtools,comment,amsthm,amsmath}
\usepackage{appendix}
\usepackage[dvipsnames]{xcolor}
\usepackage[nocompress]{cite}
\usepackage{enumerate}

\usepackage[colorlinks=true]{hyperref}
\usepackage{bookmark} %
\usepackage{xspace}
\usepackage{bm}
\usepackage[capitalize]{cleveref}

\allowdisplaybreaks

\newtheorem{theorem}{Theorem}
\newtheorem{lemma}[theorem]{Lemma}
\newtheorem{corollary}[theorem]{Corollary}
\newtheorem*{theorem*}{Theorem}
\newtheorem{definition}{Definition}

\theoremstyle{remark}

\newcommand{\threemaj}{\textsc{3-Majority}\xspace}
\newcommand{\hmaj}{$h$\textsc{-Majority}\xspace}
\newcommand{\twochoices}{\textsc{2-Choices}\xspace}
\newcommand{\averaging}{\textsc{Averaging}\xspace}
\newcommand{\voter}{\textsc{Voter}\xspace}
\newcommand{\median}{\textsc{Median}\xspace}
\newcommand{\undecidedstate}{\textsc{Undecided-State}\xspace}
\newcommand{\threedyn}{\threemaj dynamics\xspace}

\newcommand{\config}{\bm{x}}

\newcommand{\myalpha}{\textsc{alpha}\xspace}
\newcommand{\mybeta}{\textsc{beta}\xspace}

\newcommand{\biastime}[1]{s\left (#1\right )}

\newcommand{\biaseq}{s_{\text{eq}}}

\newcommand{\pr}[1]{\mathbb{P}\left [#1\right ]}
\newcommand{\expect}[1]{\mathbb{E}\left [#1\right ]}

\newcommand{\given}{\ \middle| \ }

\newcommand{\nat}{\mathbb{N}}

\newcommand{\real}{\mathbb{R}}

\newcommand{\abs}[1]{\left\lvert #1 \right\rvert}
\newcommand{\myTheta}[1]{\Theta\left (#1\right )}
\newcommand{\myOmega}[1]{\Omega\left (#1\right )}
\newcommand{\mybigo}[1]{\mathcal{O} \left (#1\right )}
\newcommand{\bigo}[1]{\mathcal{O}}

\newcommand{\filtration}[1]{\mathcal{F}_{#1}}

\newcommand{\xmin}{x_{min}}
\newcommand{\xmax}{x_{max}}

\newcommand{\diff}{\mathrm{d}}

\newcommand{\punif}{p}

\begin{document}

	\title{Phase transition of the 3-majority opinion dynamics with noisy interactions}
	
	\author{Francesco d'Amore} 
	\author{Isabella Ziccardi}
	\affil{Bocconi University, BIDSA, Italy. {\protect\\ {\tt francesco.damore2@unibocconi.it} \protect\\ {\tt isabella.ziccardi@unibocconi.it}}}
	
	\date{}
	
	\maketitle 
	
	\begin{abstract}
    Communication noise is a common feature in several real-world scenarios where systems of agents need to communicate in order to pursue some collective task. Indeed, many biologically inspired systems that try to achieve agreements on some opinion must implement \emph{resilient} dynamics, i.e. are not strongly affected by noisy communications. In this work, we study the \textsc{3-Majority} dynamics, an opinion dynamics that has been shown to be an efficient protocol for the majority consensus problem, in which we introduce a simple feature of uniform communication noise, following D'Amore et al.\ (2022). We prove that, in the fully connected communication network of $n$ agents and in the binary opinion case, the process induced by the \textsc{3-Majority} dynamics exhibits a phase transition. For a noise probability $p < 1/3$, the dynamics reaches in logarithmic time an almost-consensus metastable phase which lasts for a polynomial number of rounds with high probability. We characterize this phase by showing that there exists an attractive equilibrium value $\biaseq \in [n]$ for the bias of the system, i.e.\ the difference between the majority community size and the minority one. Moreover, we show that the agreement opinion is the initial majority one if the bias towards it is of magnitude $\Omega(\sqrt{n\log n})$ in the initial configuration. If, instead, $p>1/3$, we show that no form of consensus is possible, and any information regarding the initial majority opinion is lost in logarithmic time with high probability. 
    Despite more communications per round being allowed, the \textsc{3-Majority} dynamics surprisingly turns out to be less resilient to noise than the \textsc{Undecided-State} dynamics, whose noise threshold value is $p = 1/2$.
\end{abstract}

\noindent
{\bf Keywords:} opinion dynamics, consensus problem, randomized algorithms, distributed computing.

\section{Introduction}
The \emph{consensus problem} is a fundamental problem in distributed computing \cite{BecchettiCN20} in which we have a system of agents supporting opinions that interact between each other by exchanging messages, with the goal of reaching an agreement on some \emph{valid} opinion (i.e.\ an opinion initially present in the system). In particular, we focus on the \emph{majority consensus problem} where the goal is to converge towards the initial majority opinion. The numerous theoretical studies in this area find justifications in many different application scenarios, ranging from social networks \cite{MosselNT14,AcemogluCFO13}, swarm robotics \cite{Bayindir16}, cloud computing, communication networks \cite{RuanM08}, and distributed databases \cite{DietzfelbingerGMMPR10}, to biological systems \cite{FeinermanHK17,FraigniaudN19}. As for the latter, the goal of the majority consensus problem is to model some real-world scenarios where biological entities need to communicate and agree in order to pursue some collective task. Many biological entities in different real situations perform this type of process, e.g.\ molecules \cite{Carroll04}, bacteria \cite{Bassler02}, flock of birds \cite{Ben-ShaharDDS14}, school of fish \cite{SumpterKJCW08}, or social insects \cite{FranksPMBS02}, such as honeybees \cite{reina2017}.

In such scenarios, some form of noise often affects communication among agents. For this reason, one of the main goals of network information theory is to guarantee reliable communications in noisy networks \cite{GK2011}. 
In this context, error-correcting codes are very effective methods to reduce communication errors in computer systems \cite{KoetterK08,Todd05}, and this is why many theoretical studies of the (majority) consensus problem assume that communication between entities occurs without error, and instead consider some adversarial behavior (e.g.\ byzantine fault \cite{BecchettiCNPT16}). 
Despite their effectiveness in computer applications, error-correcting codes are quite useless if we want to model consensus in biological systems. 
Indeed, they involve sending complicated codes through communication links, and it is reasonable to assume that biological type entities communicate between each other in a simpler way.
For this reason, in recent years many works have been focusing on the study of opinion dynamics where the communication between entities is unreliable and subjected to uniform noise \cite{FeinermanHK17,FraigniaudN19,dAmore2022undecidedSI,CrucianiMQR21}.

Opinion dynamics do not have a formal definition, but we provide an informal one for the sake of the intuition, following \cite{natale2017}.
\begin{definition}[Opinion dynamics -- informal]\label{def:intro:dynamics}
    An opinion dynamics is a synchronous distributed algorithm characterized by a very simple structure. In this structure, the state of a node at round \(t\) depends only on its own state and a symmetric function of the multiset of states of its neighbors at round \(t-1\). 
    This function defines the update rule and can be randomized.
    Furthermore, the update rule is the same for every graph and every node and does not change over time. 
\end{definition}
However, within the  constraints of \cref{def:intro:dynamics}, one can still devise very complex updating rules, which would violate the spirit of simplicity that characterizes the concept of opinion dynamics.

The first consensus dynamics (i.e.\
 opinion dynamics for the consensus problem) that have been studied in the presence of noise communication are linear opinion dynamics, such as the \voter dynamics and the \averaging dynamics.
In these dynamics, the function describing a single agent's state update rule is linear with respect to the states of its neighbors.
Interestingly, they were studied in the presence of uniform noise communication \cite{lin2007} or of \textit{stubborn agents} (i.e.\ agents that never change opinion)  \cite{mobilia2003,mobilia2007,yildiz2013}.
In these settings, only \emph{metastable} forms of consensus can be achieved, where a large subset of the agents agree on an opinion while other opinions remain supported by smaller subsets of agents.
However, the \voter model has a slow convergence time even in fully connected networks and a large initial bias towards some majority opinion \cite{hassin2001}, and the \averaging dynamics requires agents to perform non-trivial computation and, more importantly, to have large local memory. For these reasons, linear opinion dynamics struggle explaining the observed metastable consensus in multi-agent systems  \cite{boczkowski2019,condon2019,natale2017},
and many research papers have begun to investigate new, more plausible, non-linear opinion dynamics.

In this work, we consider the \threemaj dynamics, which is based on a majority rule, the latter being widely employed also in the biological research field \cite{dong2010,Chaouiya2013}. In particular, we introduce in the system a uniform communication noise feature, following the definition of \cite{dAmore2022undecidedSI}. 
Such dynamics, without communication noise, has a similar behavior to the \undecidedstate dynamics \cite{BecchettiCN20}. As we describe in the next section, the two dynamics behave similarly even in the presence of uniform noise, %
as both exhibit a phase transition, but with different thresholds. Indeed, although the \threemaj dynamics makes use of more per-round communications, it turns out to be less resilient to noise than the \undecidedstate dynamics.

\subsection{Our result}

In this work, we study the \threemaj dynamics over a complete network of $n$ agents holding opinions from a binary set $\Sigma = \{\myalpha,\mybeta\}$. In each round, each agent pulls the opinions of three neighboring agents chosen independently uniformly at random and updates its opinion to the majority of the three. This dynamics is shown to be a fast, robust protocol for the majority consensus problem in different network topologies, ranging from complete graphs to sparser graphs \cite{BecchettiCN20}. For further details on previous results about the \threemaj dynamics, we defer the reader to \cref{sec:related}.

In this work, we introduce to the \threemaj dynamics an \emph{uniform communication noise} feature, following the definition in \cite{dAmore2022undecidedSI}, for which we give an equivalent formulation.
Each communication with a sampled neighbor is noisy with probability $\punif \in (0,1)$, i.e.\ with probability $p$ the received opinion is sampled u.a.r.\ in the opinion set  $\Sigma$. Instead, with probability $1-\punif$ the communication is unaffected by noise.

Even if considering the complete graph is a strong assumption for such communication networks, we note that, in every round, an agent pulls an opinion from three neighbors, resulting in a round-by-round communication pattern that forms a dynamic graph with $\mybigo{n}$ edges. Furthermore, this model can capture the behavior of bio-inspired multi-agent systems where mobile agents meet randomly at a relatively high rate. For more details about models for bio-inspired swarms of agents, we refer to  \cite{valentini2017}.

To keep track of the progress of the dynamics towards consensus, we describe the dynamics via the bias at time $t$, denoted by $s_t$, which represents the difference between the sizes of the majority and minority opinion communities at time $t$. Note that the protocol, in the case of a complete graph of fixed size \(n\) with binary opinions, is completely described by $\{s_t\}_t$.
We prove that, in the aforementioned setting, the \threemaj dynamics exhibit a phase transition. Our results are summarized in the following theorem.

\begin{theorem*}
Let $\{s_t\}_{t \geq 0}$ be the bias of the process induced by the \threemaj dynamics
with uniform noise probability $p$, and let $s_0$ the initial value of the bias. The following statements hold.
\begin{enumerate}
    \item 
    If $p<1/3$ and if $s_0 = \Omega(\sqrt{n \log n})$, consider $\biaseq = \frac{n}{1-p}\sqrt{\frac{1-3p}{1-p}}$, and let $\varepsilon>0$ be any sufficiently small constant. Then, there exists   a time $\tau_1=\mathcal{O}(\log n)$ such that, w.h.p.,\footnote{An event holds \emph{with high probability} (w.h.p.\ in short) with respect to $n$ if the probability that it occurs is at least $1 - n^{-\Theta(1)}$.} the process at time $\tau_1$ reaches a metastable almost-consensus phase characterized by the equilibrium point $\biaseq$, i.e.\
    \[s_{\tau_1} \in [(1-\varepsilon)\biaseq, (1+\varepsilon)\biaseq].\]
    Moreover, the bias remains confined in such interval for $n^{\Theta(1)}$ rounds w.h.p.
    \item If $p <1/3$ and if $s_0=\mathcal{O}(\sqrt{n \log n})$, there exists a time $\tau_2=\mathcal{O}(\log n)$ such that, w.h.p., the system becomes unbalanced towards an opinion, i.e.\
    \[s_{\tau_2}=\Omega(\sqrt{n \log n}).\]
    \item If $p>1/3$ and if $s_0 = \Omega(\sqrt{n \log n})$, there exists a time $\tau_3=\mathcal{O}(\log n)$ such that, w.h.p., at time $\tau_3$ the majority opinion is lost, i.e.\ $s_{\tau_3}=\mathcal{O}(\sqrt{n})$. In addition, with constant probability, at time $\tau_3+1$ the majority opinion changes. Moreover, for $n^{\Theta(1)}$ additional rounds the absolute value of the bias is $\mathcal{O}(\sqrt{n \log n})$ w.h.p.
\end{enumerate}
\end{theorem*}

Interestingly, when \(p < 1/3\), the equilibrium point $\biaseq$ is the fixed point of the function determining the expectation of the bias at the next step, i.e.\ $\expect{s_{t+1} \mid s_t = \biaseq} = \biaseq$.

Our results show that the \threemaj dynamics is less resilient to noise compared to the \undecidedstate dynamics, despite allowing more communication per-round. Indeed, in \cite{dAmore2022undecidedSI} (conference version \cite{damore20}), it is demonstrated that the phase transition for the \undecidedstate dynamics has a threshold at \(p = 1/2\).\footnote{
We remark that in \cite{dAmore2022undecidedSI}, the threshold for the \undecidedstate dynamics is  \(p' = 1/6\) and corresponds to \(p = 1/2\) in our model. In fact, in the cited work, each sampled opinion is non-noisy with probability $1-2p'$, otherwise it turns into the opposite opinion with probability $p'$ and into the undecided state with probability $p'$. This corresponds to a framework where with probability \(1 - 3p'\) the opinion is non-noisy, while with probability \(
p = 3p'\) the received opinion is sampled u.a.r.\ among the possible ones.}
In the \undecidedstate dynamics, at each round, each agent pulls a single neighboring opinion \(x\) uniformly at random. If the agent's former opinion \(y\) differs from \(x\), the agent becomes \emph{undecided}. Once undecided, the agent will adopt any subsequent opinion it encounters.

The behavior of the \undecidedstate dynamics with noise is similar to the one of the \threemaj dynamics stated in the above theorem. If $p$ is below the threshold value,
the dynamics w.h.p.\ rapidly breaks the symmetry and converges in logarithmic time to a metastable phase of almost-consensus. If otherwise $p$ is above the threshold, no form of consensus is possible, since the bias keeps bounded by $\mybigo{\sqrt{n\log n}}$ for a polynomial number of rounds, w.h.p. In comparison to our work, \cite{dAmore2022undecidedSI} do not find the precise value at the equilibrium point, but we believe that also the \undecidedstate dynamics has an equilibrium point if $p$ is below the threshold, as the authors show in the experimental section (\cite[Figure 1E]{dAmore2022undecidedSI}).

In \cref{sec:dynamics:simulations}, we provide some experiments, which confirm our theoretical results. In addition, we simulate the \threemaj dynamics over communication graphs sampled as Erd\H{o}s Rényi random graphs with different parameters. The latter results provide experimental evidence that the \threemaj dynamics is fast and resilient even in sparser communication graphs with strong connection properties, and phase transitions similar to the fully connected case are exhibited. For this reason, we believe that sparser topologies are worth to be analyzed.
Furthermore, it would be interesting to see whether the \threemaj dynamics with an arbitrary number of possible opinions, with the same noise model, has the exact same phase transition at the noise threshold value $\punif = 1/3$. We remark that when $p = 1/3$, this corresponds to the fact that, for each node and at each round, exactly one communication among the three ones is noisy in expectation.

 \subsection{Related Works}\label{sec:related}

\paragraph{Consensus problem.} The study of the \threemaj dynamics arises on the ground of the results obtained for the \median dynamics in \cite{DoerrGMSS11}. The \median dynamics considers a totally ordered opinion set, in which each agent pulls two neighbor opinions $i,j$ u.a.r.\ and then updates its opinion $k$ to the median between $i,j$, and $k$. The  \median dynamics is a fault-tolerant, efficient dynamics for the majority consensus problem. However, as pointed out in \cite{BecchettiCN20}, the \median dynamics may not guarantee with high probability convergence to a valid opinion in case of the presence of an adversary, which is needed for the consensus problem. Moreover, the opinion set must have an ordering, property that might not be met by applicative scenarios such as biological systems \cite{BecchettiCN20}. These facts naturally lead researchers to look for efficient dynamics that satisfy the above requirements.

To the best of our knowledge, \cite{abdullah2015} is the first work analyzing the \hmaj dynamics. In detail, in the \hmaj dynamics we have $n$ nodes and, at every round, every node pulls the opinion from $h$ random neighbors and sets his new opinion to the majority one (ties are broken arbitrarily). 
More extensive characterizations of the \threemaj dynamics over the complete graph are given in \cite{BecchettiCNPST17,berenbrink2017, BecchettiCNPT16,GhaffariL18}. 

In \cite{BecchettiCNPST17}
it is shown that the \threemaj dynamics is a fast, fault-tolerant protocol for (valid) majority consensus in the case of $k \ge 2$ colors, provided that there is an initial bias towards some majority opinion. Furthermore, \cite{BecchettiCNPST17} shows an exponential time-gap between the \threemaj consensus process and the median process in \cite{DoerrGMSS11}, thus establishing its efficiency. 
In \cite{BecchettiCNPT16}, the analysis is extended to any  initial configuration in the many-color case, in the presence of a different kind of bounded adversaries. The authors of \cite{BecchettiCNPT16} emphasize how the absence of an initial majority opinion considerably complicates the analysis, in that it must be proved that the process breaks the initial symmetry despite the presence of the adversary. Indeed, before the symmetry breaking, the adversary is more likely to cause undesired behaviors. 
The strongest result about the convergence of the \threemaj is in \cite{GhaffariL18}. The authors show that in the case of $k$ opinions, the process converges in time $O(k \log n)$ rounds, and this result it is tight when $k= O(\sqrt{n})$. 
The \threemaj dynamics is also studied in different graph topologies: \cite{kang2019} analyzes the \threemaj process in graphs of minimum degree $n^\alpha$, with $\alpha = \myOmega{(\log \log n)^{-1}}$, starting from random biased binary configurations.

Another efficient opinion dynamics for the majority consensus problem
is the \twochoices dynamics, in which each agent samples two neighbors u.a.r.\ and updates its opinion to the majority opinion among its former opinion and the two sampled neighbor opinions if there is any, otherwise it keeps its opinion. We remark that the expected round-by-round behaviors of the \twochoices dynamics and that of the \threemaj are the same, while the actual behaviors differ substantially in high probability \cite{berenbrink2017}.
For example, we have run simple experiments that suggest that our uniform noise model on the \twochoices dynamics yields a threshold noise value $\punif = 1/2$, just like the \undecidedstate dynamics.

As the \twochoices and the \threemaj dynamics, the \undecidedstate dynamics is an efficient majority consensus protocol, with the difference that it requires only one communication per-round for each agent. Further description was already given in the previous section. It is worth mentioning the more recent work \cite{bankhamer2021}, which analyzes a variant of the \undecidedstate dynamics in the many-color case starting from any initial configuration. 
For an overview on the state of the art about opinion dynamics we defer the reader to \cite{BecchettiCN20}.

\paragraph{General consensus algorithms in the GOSSIP model.}
In \cite{GhaffariP16a,BerenbrinkFGK16}, the authors address the majority consensus problem in the GOSSIP model where the underlying communication network is fully 
connected and the agents may support $k$ different opinions. 
In the GOSSIP model, at each round every node samples a constant number of neighbors u.a.r.\ and exchanges \(B\) bits with it, where \(B\) usually is \( \mybigo{\log n}\).
In \cite{GhaffariP16a}, the authors propose an algorithm achieving majority consensus in $\mathcal{O}(\log k \log n)$ rounds, where the agents have local memory of size $\log k + \mathcal{O}(1)$ bits and can send messages of size $\log k + \mathcal{O}(1)$ bits. 
The majority consensus is reached provided that the initial bias is at least $\Omega( \sqrt{n \log n})$. 
In \cite{BerenbrinkFGK16}, the authors extend the result by \cite{GhaffariP16a} and provide a first algorithm achieving the majority consensus in $O(\log k \log \log n)$ rounds, where each agent uses $\log k + \Theta(\log \log k)$ bit of memory, and a second algorithm with runtime $O(\log  n \log \log n)$ which uses $\log k +4$ bit of memory for each agents. Both algorithms require an initial bias of $\Omega(\sqrt{n} \log^2 n)$.
Notice that the algorithms provided by \cite{GhaffariP16a,BerenbrinkFGK16} are more complex algorithms w.r.t.\ the consensus dynamics we mentioned before (e.g.\ \threemaj, \twochoices, \undecidedstate), as they work in \emph{phases}: the rules defining local computations and communications change over time and, hence, require higher local memory and communication bandwidth.  
Hence, such algorithms violate the constraints of \cref{def:intro:dynamics}.

\paragraph{Consensus problem in the presence of noise or stubborn agents.}
The authors of \cite{WangL09} initiate the study of the consensus problem in the presence of communication noise. They consider the Vicsek model \cite{vicsek1995}, in which they introduce a noise feature and a notion of robust consensus. Subsequently, dynamics for the consensus problem with noisy communications have received considerable attention. This direction is motivated, among many reasons, by the desire to find models for the consensus problem in natural phenomena \cite{FeinermanHK17}.

The communication noise studied in this type of problem can be divided into two types: uniform (or unbiased) and non-uniform (or biased). The uniform noise wants to capture errors in communications between agents in real-world scenarios, in which communication noise affects all opinions in the same way. The non-uniform communication noise instead describes the case in which opinions are affected differently from one another.

The authors of \cite{FeinermanHK17} are the first to explicitly focus on the uniform noise model. In detail, they study the broadcast and the majority consensus problem when the opinion set is binary. In their model of noise, every bit in every exchanged message is flipped independently with some probability smaller than $1/2$. As a result, the authors give natural protocols that solve the aforementioned problems efficiently. The work \cite{FraigniaudN19} generalizes the above study to opinion sets of any cardinality. 

As for the non-uniform communication noise case, in \cite{CrucianiMQR21} it is considered the \hmaj dynamics with a binary opinion set $\{\myalpha,\mybeta\}$, with a probability $p$ that any received message is flipped towards a \emph{fixed} preferred opinion \mybeta, while with probability $1 - p$ the former message keeps intact. They suppose there is an initial majority agreeing on $\myalpha$, and they analyze the \emph{time of disruption}, that is the time the initial majority is subverted. They prove there exists a threshold value $p^{\star}=p^{\star}(h)$ such that, if $p < p^{\star}$ the time of disruption is at least polynomial w.h.p., and if $p > p^{\star}$, the time of disruption is constant, w.h.p. Their result holds for any sufficiently dense graph.

The noise feature affecting opinion dynamics, when the underlying communication network is complete, has been shown to be equivalent to a model without noise, in which communities of stubborn agents (i.e.\ they never change opinion) are added to the network \cite{dAmore2022undecidedSI}.
For this reason, we discuss some previous works that consider such a model. 

In \cite{yildiz2013}, the authors show that, in the voter model, the presence of stubborn agents with opposite opinions precludes the convergence to consensus. The work \cite{MukhopadhayayMR20} studies the asynchronous voter rule and the asynchronous majority rule dynamics with Poisson clocks when the opinion set is binary.
The authors use mean-field techniques, and focus on two different scenarios: In the first, some agents have a probability (which depends on their current opinion) not to update when the clock ticks. In the second, there are stubborn agents. In the second case, which directly relates to our work, they show that for the \threemaj dynamics, there are either one or two possible stable equilibria, depending on the sizes of the stubborn communities, which are reached in logarithmic time. If the two sizes are close to each other and not too large, then agreement on both opinions is possible. Otherwise, either no agreement is possible, or the process converges to an agreement towards a single opinion, which is that of the largest stubborn community. The case in which the two stubborn communities have equal size corresponds to the uniform communication noise model. 

We remark, however, that in our work we consider a different setting: First, we consider the synchronous version of the \threemaj dynamics, which cannot be analyzed with the same tools of the asynchronous version. 
Furthermore, the mean-field arguments do not capture several aspects of the process, such as metastability, which in \cite{MukhopadhayayMR20} is shown only through simulations. For example, the actual behavior of the process in the long term is oscillating in a very small interval around the equilibrium values, spending long times in those intervals, and eventually switching between the two. We characterize the width of the oscillation interval and show there is high probability of convergence, providing also a lower bound on the time the process spends in the equilibrium interval.

\subsection*{Structure of the paper}
\cref{sec:preliminaries} contains the preliminaries for the analysis and the result statements. \cref{ssec:victory_majority} contains the victory of majority case, i.e.\ where $p<1/3$ and there is large enough initial bias. \cref{ssec:symmetry_breaking} contains the symmetry breaking result, i.e.\ where $p<1/3$ but we have no condition on the initial bias. In \cref{ssec:victory_noise} we analyze the victory of noise case, i.e.\ where $p>1/3$. Finally, in \cref{sec:dynamics:simulations} we show some simulation results.

\section{Preliminaries}\label{sec:preliminaries}

We consider the complete graph $G = K_n = (V,E)$ of $n$ nodes (the agents), where each node is labelled uniquely with labels in $[n] \coloneqq \{1,\dots, n\}$. 
Each node supports a binary opinion from a set of opinions $\Sigma = \{\myalpha, \mybeta\}$.
The \threemaj dynamics defines a Markov Chain $\{M_t\}_{t \in \nat}$ which is described by the opinion of the nodes at each time step, i.e.\ $M_t = (i_1(t), \dots, i_n(t)) \in \Sigma^n$ for every $t \ge 0$, where $i_j(t)$ is the opinion of node $j$ at time $t$. 
The transition probabilities are characterized iteratively by the majority update rule as follows: given any time $t \ge 0$, let $M_t \in \Sigma^n$ be the state of the process at time $t$. Then, at time $t+1$, each node $u \in V$ samples three agents independently uniformly at random (with replacement) and updates its opinion to the majority one among the sampled neighbor opinions. For the sake of clarity, we remark that when $u$ samples a neighbor node twice, the corresponding opinion counts twice.

\paragraph{Communication noise.}
To this dynamics, we introduce a uniform communication noise feature. Let $0 < \punif < 1$ be a noise probability. When a node pulls a neighbor opinion, with probability $\punif$ the received opinion is sampled u.a.r.\ in $\Sigma$ and instead, with probability $1-\punif$, the former opinion is kept intact as it is received.

\bigskip 

 Since the communication network is complete, knowledge of the graph's size implies that the process's state is fully characterized by the number of nodes supporting a given opinion. Hence, we can write $M_t=(a_t,b_t)$, where $a_t$ is the number of the nodes supporting opinion $\myalpha$ at time $t$, and $b_t$ is the analogous for opinion $\mybeta$. 

To keep track of the progress of the dynamic towards the consensus, we define the bias of the process at time $t$  as the difference between the number of agents supporting the opinion $\mybeta$ and the number of agents supporting $\myalpha$, i.e.\
\begin{align}\label{eq:s_t}
    s_t = b_t-a_t = 2b_t - n,
\end{align}
which takes value in $\{-n, \dots, n\}$.
We notice that, once $n$ is fixed, the process can be fully characterized only by the values of the bias, i.e.\ $\{s_t\}_{t \geq 0}$. We remark that $s_t > 0$ if the majority opinion at time $t$ is $\mybeta$ and $s_t<0$ if it is $\myalpha$.  We say that configurations having bias $s_t \in \{n,-n\} $ are monochromatic, meaning that every node supports the same opinion, while a configuration with $ s_t= 0 $ is symmetric.  We finally remark that the random variable $b_{t}$ (and, analogously, $a_{t}$) is the sum of i.i.d.\ Bernoulli random variables. In detail, if $X^{(t)}_i$ is the r.v.\ yielding 1 if node $i$ adopts opinion \mybeta at round $t+1$, and 0 otherwise, then $b_{t}= \sum_{i \in [n]}X_i^{(t)}$. Therefore, for \eqref{eq:s_t},
\begin{equation}
\label{eq:s_t_sum_independent}
    s_t = 2 \sum_{i \in [n]}X_i^{(t)}-n = \sum_{i \in [n]}(2X_i^{(t)}-1),
\end{equation}
where $(2X_i^{(t)}-1)$ are i.i.d.\ taking values in $\{-1,1\}$. For this reason, we can apply the Hoeffding bound (\cref{lemma:hoeffding}) to the bias.

\paragraph{Further notation.}
For any function $f(n)$, we make use of the standard Landau notation $\mathcal{O} \left( f(n) \right)$,
$ \Omega\left(f(n)\right)$,$\Theta\left(f(n)\right)$. Furthermore, for a constant $c > 0$, we write $\mathcal{O}_c\left(f(n)\right)$,$ \Omega_c\left(f(n)\right)$, and $\Theta_c\left(f(n)\right)$ if the hidden constant in the notation depends on $c$.

\bigskip

We first compute the expectation of the bias at time $t$, conditional on its value at time $t-1$.

\begin{lemma}\label{lem:expectation_bias} 	Let $\{s_t\}_{t \geq 0}$ be the process induced by the \threemaj  dynamics with uniform noise probability $p \in (0,1)$. The conditional expectation of the bias is
\begin{align}
	\expect{s_{t}\mid s_{t-1}=s}=\frac{s(1-p)}{2}\left(3-\frac{s^2}{n^2}(1-p)^2\right).
\end{align}
\end{lemma} 
\begin{proof}
    Let $b = b_t$ and $a = a_t$. Then $s = b - a$ and $n = a + b$, which implies $b=(n+s)/2$ and $a=(n-s)/2$. The probability that, when a node samples a neighbor, it receives opinion \mybeta is $b' = (b/n)\cdot (1-p) + p/2$, where $(b/n)\cdot (1-p)$ is the probability to receive a non-noisy message which contains opinion $\mybeta$, and $p/2$ is the contribution of the noise. Analogously, the probability that it receives opinion $\myalpha$ is  $a' = (a/n)\cdot (1-p) + p/2$.
	Then, the probability the node updates its opinion to $\mybeta$ is $(b')^3 + 3a'(b')^2$.
    For the sake of the calculations, we define \(s' = b' - a' = (s/n)(1-p) \).
    Since \(1 = b' + a'\), we have \(b' = (1 + s')/2\) and \(a' = (1-s')/2\).
    So, \cref{eq:s_t} implies that
	\begin{align*}
    	\expect{s_t \mid s_{t-1}=s} & = 2n\left((b')^3+3a'(b')^2\right)-n   \\
     & = n\left(2 (b')^2(b' + 3a') - 1\right) \\
     & =  \frac n2 \left((1 + s')^2(2 - s') - 2\right) \\
     & =  \frac n2  \left(3s' - (s')^3 \right) \\
     & = \frac{s(1-p)}{2}\left(3-\frac{s^2}{n^2}(1-p)^2\right).
	\end{align*}
\end{proof}

By the lemma above, we deduce that there are three fixed points of the conditional expectation of the bias in the next step. The first one corresponds to $s=0$, and the other (possible) equilibrium correspond to the condition
\[
\frac{1-p}{2}\left(3-\frac{s^2}{n^2}(1-p)^2\right) =1.
\]
The latter condition results in 
\begin{align*}
s & = \pm \frac{n}{(1-p)} \cdot \sqrt{\frac{3(1-p)-2}{(1-p)}} \\
& =  \pm \frac{n}{(1-p)} \cdot \sqrt{\frac{1-3p}{1-p}},
\end{align*}
which is well defined if only if $ p \le 1/3$. We will denote the absolute value of the latter two values by $\biaseq$.

\section{Victory of the majority}\label{ssec:victory_majority}

The aim of this section is to prove the following theorem, which shows how the dynamics solves the majority consensus problem when $\punif < 1/3$ in a weak form, since only an almost-consensus is reached. 

\begin{theorem}[Victory of the majority]
	\label{thm:victory-majority}
	Let $\{s_t\}_{t \geq 0}$ be the process induced by the \threemaj  dynamics with uniform noise probability $p <  1/3$. Let $\varepsilon>0$ be any arbitrarily small constant such that $\varepsilon < 1/3$ and $\varepsilon \leq \frac{2(1-3p)}{3(1-p)}$, and
	let $\gamma > 0$ be any constant. Let $\biaseq = \frac{n}{(1-p)}\sqrt{\frac{1-3p}{1-p}}$. Then, for any starting configuration $ s_0 $ such that $s_0 \ge \gamma \sqrt{n\log n}$ and for any sufficiently large $n$, the following holds w.h.p.:
	\begin{enumerate}[(i)]
	    \item there exists a time $\tau_1 = \mathcal{O}_{\gamma,\varepsilon,p}(\log n)$ such that
    $
	    (1-\varepsilon)\biaseq \leq s_{\tau_1}\leq (1+\varepsilon)\biaseq
    $;
	\item there exists a value $c = {\Theta_{\gamma,\varepsilon,\punif}(1)}$ such that, for all $k\leq n^c$, 
	$
	    (1-\varepsilon)\biaseq \leq s_{\tau_1+k}\leq (1+\varepsilon)\biaseq
	$.
	\end{enumerate}
\end{theorem}

In each of the following statements we assume that $\{s_t\}_{t \geq 0}$ is the bias of the process induced by the \threemaj dynamics with uniform noise probability $p<1/3$.

We first show a lemma which states that, for any small constant $\varepsilon > 0$, whenever $s_{t-1} \not \in [(1-\varepsilon)\biaseq, (1+\varepsilon)\biaseq]$, then $s_{t}$ gets closer to the $\biaseq$. 
\begin{lemma}
	\label{lem:s_t-increase-decrease}
For any constant $0 
 \le \varepsilon \le \frac{2(1-3p)}{3(1-p)}$ and for any $\gamma > 0$, if $s \geq \gamma \sqrt{n \log n}$, the following statements hold
	\begin{enumerate}[(i)]
	    \item if $s \leq (1-\varepsilon)\biaseq$, then
	    $
	        \pr{s_{t}\geq (1+3\varepsilon^2/4)s\mid s_{t-1}=s}\geq 1-\frac{1}{n^{\gamma^2 \varepsilon^4/32}}\,;
	    $
	    \item if, $s \geq (1+\varepsilon)\biaseq$, then
	    $
	        \pr{s_t \leq (1-3\varepsilon^2/4)s \mid s_{t-1}=s}\geq 1-\frac{1}{n^{\gamma^2 \varepsilon^4/32}}\,.
	    $
	\end{enumerate}
\end{lemma}

\begin{proof}
    We first notice that 
	\begin{align}
    	(1-\varepsilon)\biaseq\leq \frac{n}{1-p}\sqrt{\frac{1-3p-2\varepsilon^2}{1-p}}.
    	\label{eq:cond-biaseps}
	\end{align}
    Indeed, 
    \begin{align*}
    1 - 3p -2\varepsilon^2 - (1-\varepsilon)^2(1-3p) & = \varepsilon\left(2(1-3p) - 3\varepsilon(1 - p)\right)
    \end{align*}
    is non-negative if and only if \(0  
 \le \varepsilon \le \frac{2(1-3p)}{3(1-p)}\).

	From \cref{lem:expectation_bias}, if each $s \leq (1-\varepsilon)\biaseq$, then
	\begin{align*}
    \expect{s_{t}\mid s_{t-1}=s} & = \frac{s(1-p)}{2}\left(3-\frac{s^2}{n^2}(1-p)^2\right) \\
    & \geq s \left(\frac{3 - 3p}{2}-\frac{1-3p-2\varepsilon^2}{2}\right) \\
    & = s(1+\varepsilon^2).   
	   \end{align*}
	where the inequality follows by \eqref{eq:cond-biaseps}. By \eqref{eq:s_t_sum_independent} and by the Hoeffding bound (\cref{lemma:hoeffding}), it holds that
	\[
	    \pr{s_t \leq s(1+\varepsilon^2)-s\varepsilon^2/4 \mid s_{t-1}=s}\leq e^{-s^2\varepsilon^4/(32n)}\leq e^{-\gamma^2 \varepsilon^4\log n/32}\leq \frac{1}{n^{\gamma^2 \varepsilon^4/32}}\,.
	\]
	The second inequality in the lemma follows by a symmetric argument, observing that
	\begin{align*}
    	(1+\varepsilon)\biaseq\geq \frac{n}{1-p}\sqrt{\frac{1-3p+2\varepsilon^2}{1-p}},
	\end{align*}
 Indeed, the expression
 \[
    (1+\varepsilon)^2(1-3p) - (1-3p + 2\varepsilon^2) = \varepsilon\left(2(1-3p) - \varepsilon(1 + 3p)\right)
 \]
 is non-negative if and only if $0 \le \varepsilon \le \frac{2(1-3p)}{1+3p}$, which is guaranteed since \( \varepsilon \le \frac{2(1-3p)}{3(1-p)}\). 
\end{proof}

The following lemma serves to bound how far the bias can get from the interval $ [(1+\varepsilon)\biaseq, (1-\varepsilon)\biaseq]$. 
\begin{lemma}
	\label{lem:s_t_increase-not-too-much}
For any constants $\varepsilon>0$ and $\gamma > 0$, if $s \geq \gamma \sqrt{n \log n}$, the following statements hold
	\begin{enumerate}[(i)]
	\item if  $ s \leq (1+\varepsilon)\biaseq$, then $ \pr{s_t \geq (1-\varepsilon - \varepsilon^2)s \mid s_{t-1}=s}\geq 1-\frac{1}{n^{\gamma^2 \varepsilon^2/16}}$;
	\item if $ s \geq (1-\varepsilon)\biaseq$ with $\varepsilon < 1$, then $ \pr{s_{t}\leq
	(1+ \varepsilon)s \mid s_{t-1}=s} \geq 1 -
	\frac{1}{n^{\gamma^2 \varepsilon^2 p^2}}$.
	\end{enumerate}
\end{lemma}
\begin{proof}
	The proof is similar to that of the previous lemma. 
	From \cref{lem:expectation_bias}, we get that
	\begin{align*}
    	\expect{s_t \given s_{t-1} = s} \ge s\left(1 - \varepsilon -\frac{\varepsilon^2}{2}\right),
	\end{align*}
	which follows since $s \le (1+\varepsilon)\biaseq$ by simple calculations. By using the Hoeffding bound (\cref{lemma:hoeffding}), we get 
	\[
	    \pr{s_t \le s\left(1 - \varepsilon -\frac{\varepsilon^2}{2}\right) - \frac{\varepsilon^2 \cdot s}{2} \given s_{t-1} = s } \le e^{-\frac{\gamma^2 \varepsilon^4}{16}} = \frac{1}{n^{\frac{\gamma^2 \varepsilon^2}{16}}}.
	\]
	The second claim follows symmetrically from \cref{lem:expectation_bias} by observing that, since $s \ge (1 - \varepsilon)\biaseq$
	\begin{align*}
    	\expect{s_t \given s_{t-1} = s} \le s\left(1 + (1-3p)\varepsilon\right).
	\end{align*}
	The Hoeffding bound implies
	\begin{align*}
	    \pr{s_t \ge s\left(1 + \varepsilon\right) \given s_{t-1} = s} & \le \pr{s_t \ge s\left(1 + (1-3p)\varepsilon\right) + 2p\varepsilon \cdot s \given s_{t-1} = s} \\
     & \le e^{-\gamma^2 \varepsilon^2 p^2}  = \frac{1}{n^{\gamma^2 \varepsilon^2 p^2}}.
	\end{align*}
\end{proof}

We provide another lemma to control the behavior of the bias. 
The proof consists in the application of simple concentration bounds. 
\begin{lemma}\label{lemma:vic_maj_bounds_seq}
    For any constant $k > 0$, the following statements hold:
    \begin{enumerate}[(i)]
        \item if $ s \ge \biaseq$, then $\pr{s_t \ge 2\biaseq/3 \given s_{t-1} = s} \ge 1 - 1/n^k$.
        \item if $ 0 \le s \le 2\biaseq/3$, then $\pr{s_t \le \biaseq \given s_{t-1} = s} \ge 1 - 1/n^k$.
    \end{enumerate}
\end{lemma}
\begin{proof}
    Let $k$ be any arbitrarily large constant. 
    As for (i), \cref{lem:expectation_bias} gives that
    \[
        \expect{s_t \given s_{t-1} = s} \ge s(1-p) \ge \biaseq(1-p),
    \]
    since $\biaseq \le s \le n$. Then, let $\delta = (1-3p)/3 > 0$. By using the Hoeffding bound, it holds that
    \[
        \pr{s_{t} \le \biaseq(1 - p) - \delta \cdot \biaseq \given s_{t-1} = s } \le e^{-\frac{\delta^2 \biaseq^2}{4}} \le \frac{1}{n^k},
    \]
    where the latter inequality holds since $\biaseq = \myTheta{n}$ and $\biaseq > (2k/\delta)\log n$ for a sufficiently large $n$.
    As for (ii), \cref{lem:expectation_bias} implies that
    \[
        \expect{s_t \given s_{t-1} = s} \le \frac{3s(1-p)}{2} \le \biaseq(1-p),
    \]
    which is true since $ 0 \le s \le 2\biaseq/3$.
    The Hoeffding bound then gives
    \[
        \pr{s_t \ge \biaseq(1-p) + p\biaseq \given s_{t-1} = s} \le e^{-\frac{p^2 \biaseq^2}{4}} \le \frac{1}{n^{k}},
    \]
    where the latter inequality holds since $\biaseq = \myTheta{n}$ and so $\biaseq > (2k/p)\log n$ for a sufficiently large $n$.
    
\end{proof}

We can piece together the above lemmas, which imply the following corollary.

\begin{corollary}
	\label{cor:key} For any constant $\varepsilon>0$ such that $\varepsilon < \frac{2(1-3p)}{3(1-p)}$,
	\begin{enumerate}[(i)]
		\item If $\abs{\biaseq - s} \le (\varepsilon/4)\biaseq$, then 
		\[
		    \pr{\abs{\biaseq - s_t} \le \varepsilon\biaseq \given s_{t-1}=s}\geq 1-\frac{1}{n^{\gamma^2 \varepsilon^2 p^2 / 2^5}}\,;
		\]
		\item If $(\varepsilon/4)\biaseq \le \abs{\biaseq - s} \le \biaseq / 3$, then 
		\[
		    \pr{\abs{\biaseq - s_t} \le \abs{\biaseq - s}\cdot \left(1 -  \frac{3\varepsilon^2}{2^5}\right)\given s_{t-1} = s} \ge 1 - \frac{1}{n^{\gamma^2 \varepsilon^4 p^2 / (2^{18} 3^2)}}\,.
		\]
	\end{enumerate}
\end{corollary}
\begin{proof}
    First, we prove (i). By using \cref{lem:s_t_increase-not-too-much} and the union bound, we have that
    \[
        \pr{\left(1 - \frac{\varepsilon}{4} - \frac{\varepsilon^2}{16}\right)\cdot\left(1 - \frac{\varepsilon}{4}\right)\biaseq \le s_t \le \left(1 + \frac{\varepsilon}{4}\right)\cdot\left(1 + \frac{\varepsilon}{4}\right)\biaseq \given s_{t-1} = s} \ge 1 - \frac{1}{n^{\gamma^2 \varepsilon^2 p^2 / 2^5}}\,.
    \]
    The claim follows by osberving that
    \[
        \left[\left(1 - \frac{\varepsilon}{4} - \frac{\varepsilon^2}{16}\right)\cdot\left(1 - \frac{\varepsilon}{4}\right)\biaseq , \left(1 + \frac{\varepsilon}{4}\right)\cdot\left(1 + \frac{\varepsilon}{4}\right)\biaseq\right] \subseteq \left[(1 - \varepsilon)\biaseq, (1+\varepsilon)\biaseq\right] \,.
    \]
	
	As for claim (ii), we divide the proof in two different cases.
    Suppose, first, that $2\biaseq /3 \le s \le (1 - \varepsilon/4)\biaseq$. A constant $\varepsilon/4 \le \delta \le 1/3$ exists such that $s = (1- \delta)\biaseq$. Then, from
    \cref{lem:s_t-increase-decrease,lem:s_t_increase-not-too-much}, we have that
    \[
        \pr{(1 - \delta)\left(1 + \frac{3\varepsilon^2}{2^6}\right) \biaseq \le s_t \le \biaseq \given s_{t-1} = s} \ge 1 - \frac{1}{n^{\gamma^2 \varepsilon^4 p^2/ 2^{14}}}\,.
    \]
    Notice that 
    \begin{align*}
        & \abs{\biaseq - (1 - \delta)\left(1 + \frac{3\varepsilon^2}{2^6}\right)\biaseq} = \biaseq - (1 - \delta)\left(1 + \frac{3\varepsilon^2}{2^6}\right) \biaseq\\
        = \ & \left(\biaseq - (1-\delta)\biaseq\right) \cdot \left(1 - \frac{(1-\delta)\cdot \frac{3\varepsilon^2}{2^6} \cdot \biaseq}{\delta \cdot \biaseq}\right)
        \le  (\biaseq - s) \cdot \left(1 - \frac{3\varepsilon^2}{2^5}\right)\,,
    \end{align*}
    where in the last inequality we used that $\delta \le 1/3$.
    Hence, 
    \begin{align}\label{eq:towards_eq_1}
        \pr{\abs{\biaseq - s_t} \le \abs{\biaseq - s}\cdot \left(1 - \frac{3\varepsilon^2}{2^5}\right) \given s_{t-1} = s} \ge 1 - \frac{1}{n^{\gamma^2 p^2 \varepsilon^4 / 2^{14}}}\,.
    \end{align}
    Second, suppose $(1 + \varepsilon/4)\biaseq \le s \le 3\biaseq/2$. A constant $\varepsilon/4 \le \delta \le 1/3$ exists such that $s = (1 + \delta)\biaseq$.
    From \cref{lem:s_t-increase-decrease,lem:s_t_increase-not-too-much}, it holds that
    \[
        \pr{(1+\delta)\left(1 - \delta - \delta^2\right)\biaseq \le s_t \le (1+\delta)\left(1 - \frac{3\varepsilon^2}{2^6}\right) \biaseq \given s_{t-1} = s} \ge 1 - \frac{1}{n^{\gamma^2 \varepsilon^4 p^2 / (2^{18} 3^2)}}\,,
    \]
    where we used the union bound and the fact that $\delta \le 1/3$.
    Notice that
    \begin{align*}
        \abs{\biaseq -(1+\delta)\left(1 - \delta - \delta^2\right)\biaseq } = \ & \biaseq - (1+\delta)\left(1 - \delta - \delta^2\right)\biaseq \\
        = \ & \left( (1+\delta)\biaseq - \biaseq\right)\cdot\left(\frac{(1+\delta)(\delta+\delta^2)\biaseq}{\delta \biaseq} - 1 \right) \\
        \le \ & \left( (1+\delta)\biaseq - \biaseq\right)\cdot\left(\frac{16}{9} - 1 \right) \\
        = \ & \left( (1+\delta)\biaseq - \biaseq\right)\cdot\left(1 - \frac{2}{9} \right)\,,
    \end{align*}
    where the inequality holds since $\delta \le 1/3$.
    By simple calculations, it can be seen that $(1+\delta)\left(1 - \frac{3\varepsilon^2}{2^6}\right) \ge 1$. Then, we have also that
    \begin{align*}
    \abs{\biaseq - (1+\delta)\left(1 - \frac{3\varepsilon^2}{2^6}\right) \biaseq} = \ & (1+\delta)\left(1 - \frac{3\varepsilon^2}{2^6}\right) \biaseq - \biaseq \\ 
        = \ & \left( (1+\delta)\biaseq - \biaseq\right)\cdot\left[1 - \frac{(1+\delta)\cdot \frac{3\varepsilon^2}{2^6}\biaseq}{\delta\biaseq}  \right] \\
        \stackrel{(a)}{\le} \ & \left( (1+\delta)\biaseq - \biaseq\right)\cdot\left[1 - 3\left(1+\frac{\varepsilon}{4}\right)\cdot \frac{3\varepsilon^2}{2^6}  \right] \\
        \stackrel{(b)}{\le} \ & \left( (1+\delta)\biaseq - \biaseq\right)\cdot\left[1 -  \frac{9\varepsilon^2}{2^6}  \right]\,,
    \end{align*}
    where (a) holds since $\varepsilon/4 \le \delta \le 1/3$, and (b) holds since $\varepsilon > 0$.
    Thus,
    \begin{align}\label{eq:towards_eq_2}
        \pr{\abs{\biaseq - s_t} \le \abs{\biaseq - s}\cdot \left[1 -  \frac{9\varepsilon^2}{2^6}\right]\given s_{t-1} = s} \ge 1 - \frac{1}{n^{\gamma^2 \varepsilon^4 p^2 / (2^{18} 3^2)}}\,.
    \end{align}
    Combining \cref{eq:towards_eq_1,eq:towards_eq_2}, we get that, whenever $(\varepsilon/4)\biaseq \le \abs{\biaseq - s} \le \biaseq / 3$, then 
    \begin{align*}
        \pr{\abs{\biaseq - s_t} \le \abs{\biaseq - s}\cdot \left[1 -  \frac{3\varepsilon^2}{2^5}\right]\given s_{t-1} = s} \ge 1 - \frac{1}{n^{\gamma^2 \varepsilon^4 p^2 / (2^{18} 3^2)}}\,.
    \end{align*}
\end{proof}

We are finally ready to prove the theorem.

\begin{proof}[Proof of Theorem \ref{thm:victory-majority}]

We divide the proof in different cases.
First, suppose that $(\varepsilon/4)\biaseq \le \abs{\biaseq - s} \le \varepsilon\biaseq$. Let $T_1 = n^{\gamma^2 \varepsilon^4 p^2 / (2^{19}3^2) }$. Then, from \cref{cor:key}.(i) and (ii), and by using the chain rule, we have that 
\[
    \pr{\bigcap_{k = 1}^{T_1} \left\{\abs{\biaseq - s_{t + k}} \le \varepsilon\biaseq\right\} \given s_{t} = s} \ge 1 - \frac{1}{n^{\gamma^2 \varepsilon^4 p^2 / (2^{20}3^2) }}.
\]
This proves statement (ii) of the theorem.

Second, suppose that $\varepsilon\biaseq \le \abs{\biaseq - s} \le \biaseq/3$. Then, from \cref{cor:key}.(ii) and by using the chain rule, a time $T_2$ exists, such that 
\[
T_2 = \mybigo{- \frac{\log n}{\log \left( 1 - \frac{3\varepsilon^2}{2^5}\right)}}  = \mybigo{\log n / \varepsilon^2}
\]
such that 
\[
    \pr{\abs{\biaseq - s_{t + T_2}} \le \varepsilon\biaseq \given s_t = s} \ge  1 - \frac{1}{n^{\gamma^2 \varepsilon^4 p^2 / (2^{20}3^2) }}.
\]

Third, suppose that $s \le 2\biaseq / 3$. From \cref{lem:s_t-increase-decrease}.(i) and \cref{lemma:vic_maj_bounds_seq}.(ii), by the chain rule and the union bound, there is a time 
    \[
        T_3 = \mybigo{\frac{\log n}{\log \left(1 + \frac{3\varepsilon^2}{4}\right)}} = \mybigo{\log n / \varepsilon^2}
    \]
    such that 
    \[
        \pr{2\biaseq/3 \le s_{t + T_3} \le \biaseq \given s_t = s} \ge 1 - \frac{1}{n^{\gamma^2 \varepsilon^4 / 2^6}}\,.
    \]
Then, we are in one of the first two cases, and we conclude by using the chain rule.

Fourth, suppose that $s \ge (1 + \frac{1}{3})\biaseq $. From \cref{lem:s_t-increase-decrease}.(ii) and \cref{lemma:vic_maj_bounds_seq}.(i), and by using the chain rule, a time $T_4$ exists, with $T_4 = \mybigo{\log n}$, such that 
\[
    \pr{\abs{\biaseq - s_{T_4}} \le \biaseq/3 \given s_{t} = s} \ge 1 - \frac{1}{n^{\gamma^2/(3^4 2^6)}}\,.
\]
Statement (i) of the theorem follows by setting $\tau_1 = {T_2 + T_3 + T_4}$.
\end{proof}

\section{Symmetry breaking}\label{ssec:symmetry_breaking}
The following theorem shows how the dynamics quickly break the initial symmetry. Combining this result with \cref{thm:victory-majority}, it shows that the consensus problem is solved. 

\begin{theorem}[Symmetry breaking]\label{thm:symbreak} 
Let $\{s_t\}_{t \geq 0}$ be the process induced by the \threemaj  dynamics with uniform noise probability $p <  1/3$, and let $\gamma > 0$ be any constant. Then, for any starting configuration $s_0$ such that $|s_0| \leq \gamma \sqrt{n \log n}$ and for any sufficiently large $n$, w.h.p. there exists a time $\tau_2=\mathcal{O}_{\gamma,\punif}(\log n)$ such that $|s_{\tau_2}|\geq \gamma \sqrt{n \log n}$.
\end{theorem}

In each of the following statements, we assume that $\{s_t\}_{t \geq 0}$ is the process induced by the \threemaj dynamics with uniform noise probability $p<1/3$. The symmetry breaking analysis essentially relies on the following lemma, which has been proved in  \cite{clementi2018}.

\begin{lemma}\label{lemma:symbreak-generic}
	Let $\{X_{t}\}_{t\in \nat}$ be a Markov Chain with finite-state space $\Omega$ and let $f:\Omega\mapsto[0,n]$ be a function that maps states to integer values. Let $c_3$ be any positive constant and let $m = c_3\sqrt{n}\log n$ be a target value. Assume the following properties hold:
	\begin{enumerate}[(i)]
		\item for any positive constant $h$, a positive constant $c_1 < 1$ (which depends only on $h$) exists, such that for any $x \in \Omega : f(x) < m$,
		\[
		\pr{f(X_t) < h\sqrt{n} \given X_{t-1} = x} < c_1;
		\]
		\item there exist two positive constants $\delta$ and $c_2$ such that for any $x \in \Omega: h\sqrt{n} \leq f(x) < m$,
		\[
		\pr{f(X_t) < (1+\delta)f(X_{t-1}) \given X_{t-1} = x} < e^{-c_2f(x)^2/n}.
		\]
	\end{enumerate}
	Then the process reaches a state $x$ such that $f(x) \ge m$ within $\mathcal{O}_{c_2,\delta,c_3}(\log n)$ rounds with probability at least $1 - 2/n$.
\end{lemma}

Our goal is to apply the above lemma to the \threemaj process, which defines a Markov chain. In particular, we claim the hypothesis of \cref{lemma:symbreak-generic} are satisfied when the bias of the system is $ o\left(\sqrt{n\log n}\right) $, with $ f(\config) = \biastime{\config} $, $ m = \gamma \sqrt{n} \log n $ for any constant $\gamma > 0$. Then, \cref{lemma:symbreak-generic} implies the process reaches a configuration with bias greater than $ \myOmega{\sqrt{n\log n}} $ within time $ \mybigo{\log n} $, w.h.p. We need to prove that the two hypotheses hold.

\begin{lemma}\label{lemma:symbreak-bias-satisfies-generic}
For any constant $c_3>0$, let $s$ be a value such that $\abs{s} <c_3 \sqrt{n}\log n$. Then, 
	\begin{enumerate}[(i)]
		\item For any constant $h > 0$, there exists a positive constant $ c_1 < 1 $ depending only on $h$, such that 
		\[
			\pr{s_{t} < h \sqrt{n} \given s_{t-1} =s } < c_1;
		\]
		\item Two positive constants $ \delta, c_2$ exist (depending only on $p$), such that if $ |s| \ge h\sqrt{n}  $, then
		\[
			\pr{s_{t} < (1+ \delta)s \given s_{t-1}=s} < e^{-\frac{c_2 s^2}{n}}.
		\]
	\end{enumerate}
\end{lemma}
\begin{proof}
	As for the first claim, a simple domination argument implies that 
	\begin{equation}
	    \label{eq:simbreak_prob1}	\pr{|s_{t}| < h\sqrt{n}\given s_{t-1} = s} \le \pr{|s_{t}| < h\sqrt{n}\given s_{t-1} = 0}.
	\end{equation}
	As shown in \cref{sec:preliminaries}, $s_{t}$ is a sum of $n$ i.i.d.\ Rademacher r.v.s with zero mean and unitary variance. We can hence make use of the \cref{lemma:berry-eseen} (Berry-Essen inequality). In particular, let $\Phi(x)$ be the cumulative function of a standard normal distribution. A constant $C > 0$ exists such that
	\begin{align*}
		\abs{\pr{s_{t} \le h\sqrt{n} \given s_{t-1} = 0} - \Phi(h)} \le \frac{C}{\sqrt{n}}.
	\end{align*}
	Since $ \Phi(h) = c $ for some constant $c>0$ which depends only on $h$, we have that 
	\begin{align*}
		c - \frac{C}{\sqrt{n}} \le \pr{s_{t} \le h\sqrt{n} \given s_{t-1} = 0} \le c + \frac{C}{\sqrt{n}}.
	\end{align*}
	Since $ \pr{|s_{t}| < h \sqrt{n} \given s_{t-1} = 0} \le \pr{s_{t} \le h\sqrt{n} \given s_{t-1} = 0 }$, for $n$ large enough we get \[ \pr{|s_{t}| < h\sqrt{n} \given s_{t-1} = 0} < 2c .\] By setting $ c_1 = c/2 $, we get claim (i) from \cref{eq:simbreak_prob1}.
	
	As for the second claim, assume $s > 0$ and $h\sqrt{n}\leq s \leq h\sqrt{n} \log n$. By Lemma \ref{lem:expectation_bias} and the fact that $h\sqrt{n} \le s \le h\sqrt{n} \log n \le (1-\sqrt{\varepsilon})s_{eq}$, we have (as in Lemma \ref{lem:s_t-increase-decrease})
	\begin{align*}
	\expect{s_{t}\mid s_{t-1}=s} & = \frac{s(1-2p)}{2}\left(3-\frac{s^2}{n^2}(1-2p)^2\right) \\
 & \geq s \left(\frac{3}{2}-3p-\frac{1-6p-2\varepsilon}{2}\right) \\
 & = s(1+\varepsilon).
	\end{align*}
	By the Hoeffding bound (\cref{lemma:hoeffding}), we get that
	\begin{align*}
		\pr{s_{t} \le s \left(1+\varepsilon\right)-s\varepsilon/4 \given s_{t-1} = s} 		& \le e^{-s^2 \varepsilon^2/(32n)} .
	\end{align*}
	Observe that $ \pr{|s_{t}| \le s \left(1+3\varepsilon/4\right) \given s_{t-1} = s} \le \pr{s_{t} \le s \left(1+3\varepsilon/4\right) \given s_{t-1}=s} $. Thus, we have the claim by setting $ \delta = 3\varepsilon/4  $ and $ c_2 = \varepsilon^2 / 32 $. 
\end{proof}

The symmetry breaking is then a simple consequence of the above Lemma. 
\begin{proof}[Proof of Theorem \ref{thm:symbreak}]
    	Apply \cref{lemma:symbreak-bias-satisfies-generic,lemma:symbreak-generic} with $ h = c_3 = \gamma$.
\end{proof}

\section{Victory of noise}\label{ssec:victory_noise}

The following theorem shows that no form of consensus is possible when $\punif > 1/3$.

\begin{theorem}[Victory of noise]\label{thm:victory_noise} Let $\{s_t\}_{t \geq 0}$ be the process induced by the \threemaj  dynamics with uniform noise probability $p > 1/3$. Let $\varepsilon>0$ be any arbitrarily small constant such that $ \varepsilon < \min \{1/4, (1-p), (3p-1)/2\}$ and let $\gamma>0$ be any positive constant. Then, for any starting configuration $s_0$ such that $|s_0| \geq \gamma \sqrt{n \log n}$ and for any sufficiently large $n$, the following holds w.h.p.:
\begin{enumerate}[(i)]
    \item there exists a time $\tau_3=\mathcal{O}_{\varepsilon,p}(\log n)$ such that $s_{\tau_3}=\mathcal{O}_{\varepsilon}(\sqrt{n})$ and, moreover, the majority opinion switches at the next round  with probability $\Theta_{\varepsilon}(1)$;
    \item there exists a value $c={\Theta_{\gamma,\varepsilon}(1)}$ such that, for all $k \leq n^c$, it holds that $|s_{\tau_3+k}|\leq \gamma \sqrt{n \log n}$.
\end{enumerate}
\end{theorem}

In each of the following statements, we assume that $\{s_t\}_{t \geq 0}$ is the process induced by the \threemaj dynamics with uniform noise probability $p>1/3$. 

We apply tools from drift analysis (\cref{lemma:drift_analysis}) to the absolute value of the bias of the process, showing that it reaches magnitude $\mybigo{\sqrt{n}}$ quickly. Then, since the standard deviation of the bias is $\myTheta{\sqrt{n}}$, we have  that the majority opinion switches with constant probability (\cref{lemma:vic_noise_maj_switches}). Finally, with \cref{lemma:3maj-highp-biasdecreases}, we show that the bias keeps bounded in absolute value by $\mybigo{\sqrt{n\log n}}$.

\begin{lemma}\label{lemma:vicnoise_unsignedbias}
For any constant $\varepsilon>0$ such that $\varepsilon <(1-p)$, if $s \ge 2 \sqrt{n} / \left(\varepsilon^2\right)$, the following holds 
\[
    \expect{\abs{s_t} \given s_{t-1} = s} \le \expect{s_t \given s_{t-1} = s} \cdot \left( 1 + \frac{\varepsilon}{2}\right)\,.
\]
\end{lemma}
\begin{proof}
    It holds that
    \begin{align*}
        \abs{s_t} \le \abs{s_t - \expect{s_t \given s_{t-1} = s}} + \abs{ \expect{s_t \given s_{t-1} = s}} \,.
    \end{align*}
    Furthermore, from \cref{lem:expectation_bias}, it follows  that $\expect{s_t \given s_{t-1} = s}  \ge 0$ as long as $s \ge 0$. By writing  
    \[
        \abs{s_t - \expect{s_t \given s_{t-1} = s}} = \sqrt{\left( s_t - \expect{s_t \given s_{t-1} = s}\right)^2},
    \]
    and by using the Jensen's inequality, it follows that
    \begin{align}\label{eq:vic_noise_abs_bias}
        \expect{\abs{s_t} \given s_{t-1} = s}  & \le \sqrt{\expect{\left(s_t - \expect{s_t \given s_{t-1} = s}\right)^2 \given s_{t-1} = s}} + \expect{s_t \given s_{t-1} = s} \nonumber\\
        & = \sigma\left(\abs{s_t} \given s_{t-1} = s\right) + \expect{s_t \given s_{t-1} = s},
    \end{align}
    where $\sigma(x)$ represents the standard deviation of a r.v.\ $x$. As pointed out in \cref{sec:preliminaries}, the bias can be written as the sum of i.i.d.\ random variables $Y_{i}^{(t)}$ taking values in $\{-1, + 1\}$. For such sum of variables, the variance is linear:
    \begin{align*}
        \sigma\left(\abs{s_t} \given s_{t-1} = s\right)^2 = \sum_{i = 1}^n \sigma\left(Y_i^{(t)} \given s_{t-1} = s \right)^2 \le n\,,
    \end{align*}
    where the latter inequality holds since $\sigma\left(Y_i^{(t)} \given s_{t-1} = s \right)^2 \le 1$ for every $i$.
    From \cref{lem:expectation_bias}, we deduce that
    \[
        \expect{s_t \given s_{t-1} = s} \ge \frac{s(1-p)(3-(1-p)^2)}{2} \ge s(1-p).
    \]
    Since $s \ge \frac{2\sqrt{n}}{\varepsilon^2} \ge \frac{2\sqrt{n}}{\varepsilon(1-p)}$, we get that $\expect{s_t \given s_{t-1} = s} \ge \frac{2\sqrt{n}}{\varepsilon}$. 
    Combining the latter facts with \eqref{eq:vic_noise_abs_bias}, we obtain
    \begin{align*}
        \expect{\abs{s_t} \given s_{t-1}} & \le \expect{s_t \given s_{t-1} = s}\cdot \left( 1 +\frac{\sigma\left(\abs{s_t} \given s_{t-1} = s\right)}{\expect{s_t \given s_{t-1} = s}}\right)  \\
        & \le \expect{s_t \given s_{t-1} = s}\cdot \left( 1 +\frac{\sqrt{n}}{\frac{2\sqrt{n}}{\varepsilon}}\right) \\
        & \le \expect{s_t \given s_{t-1} = s}\cdot \left( 1 +\frac{\varepsilon}{2}\right).
    \end{align*}
\end{proof}

With next lemma, we show that the absolute value of the process quickly becomes of magnitude $\mybigo{\sqrt{n}}$.
\begin{lemma}\label{lemma:vicnoise_martingale} For any constant $\varepsilon>0$ such that $\varepsilon < \min \{(1-p), (3p-1)/2\}$ we define $s_{\text{min}} =  \sqrt{n} / \varepsilon^2$. Then, for any starting configuration $s_0$ such that $s_0 \ge s_{\text{min}}$, with probability at least $1-1/n$ there exists a time $\tau = \mathcal{O}_\varepsilon(\log n)$ such that $\abs{s_{\tau}} \le s_{\text{min}} $.
\end{lemma}
\begin{proof}
    Let $h(x) = \frac{\varepsilon \cdot x}{2}$ be a function.
    Let $X_t = \abs{s_t}$ if $s_t \ge s_{\text{min}}$, otherwise $X_t = 0$. We now estimate $\expect{X_t - X_{t - 1}\given X_{t-1} \ge s_{\text{min}}, \mathcal{F}_{t-1}}$, where $\mathcal{F}_t$ is the natural filtration of the process $X_t$. We have that
    \begin{align*}
        & \expect{X_t - X_{t-1} \given X_{t-1} \ge s_{\text{min}}, \mathcal{F}_{t-1}} \\
        = \ & \expect{X_t \given X_{t-1} \ge s_{\text{min}}, \mathcal{F}_{t-1}} - X_{t-1} \\
        \stackrel{(a)}{\le} \ & \expect{\abs{s_t} \given s_{t-1} \ge s_{\text{min}}, \mathcal{F}_{t-1}} - s_{t-1}         \\
        \stackrel{(b)}{\le} \ & \expect{s_t \given s_{t-1} \ge s_{\text{min}}, \mathcal{F}_{t-1}} \cdot \left(1 + \frac{\varepsilon}{2}\right) - s_{t-1} \\
        \stackrel{(c)}{\le} \ & s_{t-1}(1-\varepsilon)\left(1 + \frac{\varepsilon}{2}\right) -s_{t-1}  \\
        \le \ & -\frac{\varepsilon \cdot s_{t-1}}{2},
    \end{align*}
    where (a) holds because $X_t \le \abs{s_t}$, (b) follows from \cref{lemma:vicnoise_unsignedbias}, and (c) from \cref{lem:expectation_bias}.
    Thus,
    \[
        \expect{X_{t-1} - X_{t} \given X_{t-1} \ge s_{\text{min}}, \mathcal{F}_{t-1}} \ge h\left(X_{t-1}\right)\,.
    \]
    Since $h'(x) = \varepsilon/2 > 0$, we can apply \cref{lemma:drift_analysis}.(iii).
    Let $\tau$ be the first time $X_t =0$ or, equivalently, $\abs{s_t} < s_{\text{min}}$. Then
    \begin{align*}
        \pr{\tau > t \given s_0} & < \exp \left[ -\frac{\varepsilon}{2} \cdot \left(t - \frac{2}{\varepsilon} - \int_{s_{\text{min}}}^{s_0} \frac{2}{\varepsilon \cdot y} \ \diff y \right) \right] \\
        & \le \exp \left[ -\frac{\varepsilon}{2} \cdot \left(t - \frac{2}{\varepsilon} - \int_{s_{\text{min}}}^{n} \frac{2}{\varepsilon \cdot y} \ \diff y \right) \right] \\
        & = \exp \left[ -\frac{\varepsilon}{2} \cdot \left(t - \frac{2}{\varepsilon} - \frac{2}{\varepsilon}(\log n - \log s_{\text{min}})  \right) \right] \\
        & = \exp \left[ -\frac{\varepsilon}{2} \cdot \left(t - \frac{2}{\varepsilon} - \frac{2}{\varepsilon}\left((\log n)/2 + 2\log \varepsilon \right)  \right) \right]  \\&\le \exp \left[ -\frac{\varepsilon \cdot t}{2} + 1 +\frac{\log n}{2}  \right] \,.
    \end{align*}
    If $t = 4(\log n)/\varepsilon$, then we get that 
    $\pr{\tau > t \given s_0} < e^{-3(\log n)/2 + 1} < 1/n$.
\end{proof}

Next lemma states that, whenever the absolute value of the bias is of order of $\mybigo{\sqrt{n}}$, then the majority opinion switches at the next round with constant probability.

\begin{lemma}\label{lemma:vic_noise_maj_switches}
    For any constant $\varepsilon>0$ such that $\varepsilon < 1/4$, and let $s_{t-1}$ be a configuration such that $\abs{s_{t-1}} = s \le \sqrt{n} / \varepsilon$. Then, the majority opinion switches at the next round with constant probability.
\end{lemma}
\begin{proof}
    Without loss of generality, we assume that $s_{t-1} > 0$. Now, $s_{t-1} = b_{t-1} - a_{t-1}$, with $n/2 < b_{t-1} \le n/2 + \sqrt{n}/(2\varepsilon)$ and $n/2 - \sqrt{n}/(2\varepsilon) \le  a_{t-1}  < n/2 $. Both $b_{t-1}$ and $a_{t-1}$ can be expressed as the sum of i.i.d.\ Bernoulli r.v.s.
    Since $\expect{a_t \given n/2 - \sqrt{n}/(2\varepsilon) \le  a_{t-1}  < n/2} \le n/2$, we have 
    \begin{align*}
        \pr{a_t \ge \frac{n}{2} + \frac{\sqrt{n}}{2\varepsilon}  \given s_{t-1} = s } & = \pr{a_t \ge \frac{n}{2} \cdot \left(1 + \frac{1}{\varepsilon\sqrt{n}}\right) \given s_{t-1} = s}  \ge e^{-\frac{9}{2\varepsilon^2}},
    \end{align*}
    where the latter inequality holds by using the reverse Chernoff bound (\cref{lemma:chernoff-reverse}), whose hypothesis is satisfied since $\varepsilon < 1/4$.
    Thus, there is at least constant probability that the majority opinion switches.
\end{proof}

Next lemma shows that the signed bias decreases each round. 
\begin{lemma}\label{lemma:3maj-highp-biasdecreases} For any constant $\varepsilon>0$ such that $\varepsilon \leq (3p-1)/2$, the following statements hold
	\begin{enumerate}[(i)]
		\item if $ s \ge \frac{\gamma}{2}\sqrt{n\log n} $, then $\pr{s_{t}\leq (1-3\varepsilon/4)s \mid s_{t-1} = s}\geq 1-\frac{1}{n^{\gamma^2\varepsilon^2/2^7}}\,;$
		\item if $s\geq 0$, then $\pr{-\frac{\gamma}{2}\sqrt{n \log n} \leq s_t \leq s+\frac{\gamma}{2}\sqrt{n \log n} \mid s_{t-1}=s} \geq 1-\frac{2}{n^{\gamma^2/8}}.$
	\end{enumerate}
\end{lemma}
\begin{proof}
From \cref{lem:expectation_bias}, for each $s \geq 0$ it holds that
\begin{align}
    \expect{s_t \mid s_{t-1}=s} \leq  \frac{3s(1-p)}{2} \leq (1-\varepsilon)s,
    \label{exp:victory_noise1}
\end{align}
where the second inequality is true since $\varepsilon \leq (3p-1)/2$. We now apply the Hoeffding bound (\cref{lemma:hoeffding}) to $s_t$:
\begin{align*}
    \pr{s_t \geq (1-\varepsilon)s+\varepsilon\cdot s/4} \leq e^{-s^2\varepsilon^2/(32n)} \leq e^{-\gamma^2\varepsilon^2\log n/2^7} \leq \frac{1}{n^\frac{\gamma^2 \varepsilon^2}{2^7}}.
\end{align*}
As for the second claim, we notice that, from \cref{exp:victory_noise1}, $\expect{s_t \mid s_{t-1}=s}\leq s$. The Hoeffding bound (\cref{lemma:hoeffding}) now implies that
\begin{align*}
    \pr{s_t \geq s+\frac{\gamma}{2}\sqrt{n \log n}} \leq e^{-\gamma^2\log n/8}\leq \frac{1}{n^{\gamma^2/8}}.
    \label{eq:claim_2_part1_noise}
\end{align*}
Moreover, from  \cref{lem:expectation_bias}, for any $0 \le s \le n$, 
$
    \expect{s_t \mid s_{t-1} = s} \geq 0
$.
Applying again the Hoeffding bound, we get that
\begin{align*}
    \pr{s_t \geq -\frac{\gamma}{2}\sqrt{n \log n}\mid s_{t-1}=s} \leq e^{-\gamma^2 \log n/8} \leq \frac{1}{n^{\gamma^2/8}},
\end{align*}
By the union bound, we get the second claim.
\end{proof}

We are ready to prove \cref{thm:victory_noise}.
\begin{proof}[Proof of  \cref{thm:victory_noise}]
    Claim (i) follows directly from \cref{lemma:vicnoise_martingale,lemma:vic_noise_maj_switches}. 
    As for claim (ii), whenever the bias at some round $t = \tau + k$ becomes $\abs{s_{t}} \ge (\gamma/2)\sqrt{n\log n}$, from \cref{lemma:3maj-highp-biasdecreases}.(ii) (and its symmetric statement), we have that $\abs{s_t} \le \gamma\sqrt{n\log n}$ with probability $ 1 - 2/n^{\frac{\gamma^2}{8}}$. Then, from \cref{lemma:3maj-highp-biasdecreases}.(i) it follows that the bias starts decreasing each round with probability $1 - 1/n^{\gamma^2 \varepsilon^2 / 2^7}$ until reaching $(\gamma/2) \sqrt{\log n}$. This phase in which the absolute value of the bias keeps bounded by $\abs{\gamma\sqrt{n\log n}}$ lasts for at least $n^{\gamma^2 \varepsilon^2 / 2^8}$ with probability at least $1 - 1/(2n^{\gamma^2 \varepsilon^2 / 2^8})$ by using the chain rule.
\end{proof}

\section{Experiments}\label{sec:dynamics:simulations}
In this section, we describe the experiments we conducted on the \threedyn to show, in practice, the behavior of the bias for different input sizes and noise parameters. 

In \cref{fig:experiments:vic-majority:-erdos} we show the average convergence time of the \threedyn to almost-consensus in s and also Erdös-Rényi graphs \(G_{n,q}\) (with \(q = 1.25 \cdot \ln(n)/n\) to ensure connectivity with high probability).
Interestingly, even though our theorems give asymptotic bounds, the convergence time behaves as expected even for small input sizes (e.g.\ \(n = 2^{10}\) nodes). 
Furthermore, the asymptotic convergence time to almost-consensus seems to be logarithmic even in the Erdös-Rényi graph.

We also test the dynamics against expanders that have low degree: random regular graphs with degrees \(d \in \{3,5\}\) (\cref{fig:experiments:vic-majority:rand-reg}).
It seems that the convergence time is still logarithmic but quantitatively very different from the case of s: e.g.\ in \cref{fig:experiments:vic-majority:rand-reg:b} the convergence time for the random regular graph of degree \(d = 5\) with \(
2^{16}
\) nodes is around \(700\), whereas for s and Erdös-Rényi graphs it is between \(
60
\) and \(80\) (\cref{fig:experiments:vic-majority:-erdos:b,fig:experiments:vic-majority:-erdos:d}).
Notably, we also found out that the dynamics over the random regular graph with \(d = 3\) doesn't seem to converge when the noise parameter is any value \(p \ge 1/5\) (\cref{fig:experiments:diff-noise:rand-reg:d}). 
This phenomenon  suggests that the noise values determining the phase-transition can depend on the expansion and/or sparsity of the underlying graph:
these values seem to be reduced whenever the expansion of the graph decreases and its sparsity increases, as noticeable in \cref{fig:experiments:diff-noise}.
For future work, it would be interesting to study the relation between expansion/sparsity of the graph and the behavior of the \threedyn.

\begin{figure}[ht]
    \centering
    \begin{subfigure}{0.4\textwidth}
        \centering
        \includegraphics[scale=.35]{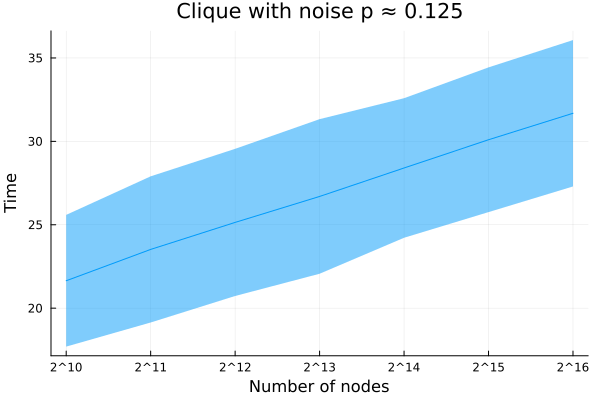}
        \caption{Average convergence time to almost-consensus in a  with noise parameter \(p = 1/8\).}
        \label{fig:experiments:vic-majority:-erdos:a}
    \end{subfigure}
    \hspace*{1cm}
    \begin{subfigure}{0.4\textwidth}
        \centering
        \includegraphics[scale=.35]{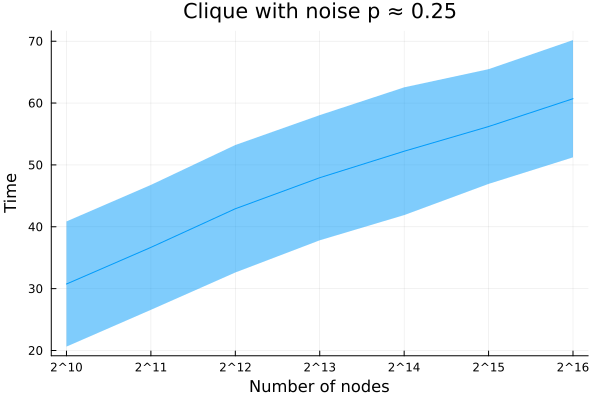}
        \caption{Average convergence time to almost-consensus in a  with noise parameter \(p = 1/4\).}
        \label{fig:experiments:vic-majority:-erdos:b}
    \end{subfigure}
    \\
    \bigskip
    \begin{subfigure}{0.4\textwidth}
        \centering
        \includegraphics[scale=.35]{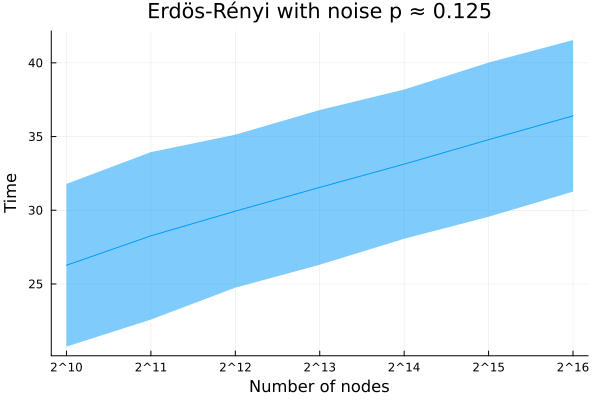}
        \caption{Average convergence time to almost-consensus in an Erdos-Rényi graph \(G_{n,q}\) with \(q = 1.25 \cdot \ln(n)/n \) and noise parameter \(p = 1/8\).}
        \label{fig:experiments:vic-majority:-erdos:c}
    \end{subfigure}
    \hspace*{1cm}
    \begin{subfigure}{0.4\textwidth}
        \centering
        \includegraphics[scale=.35]{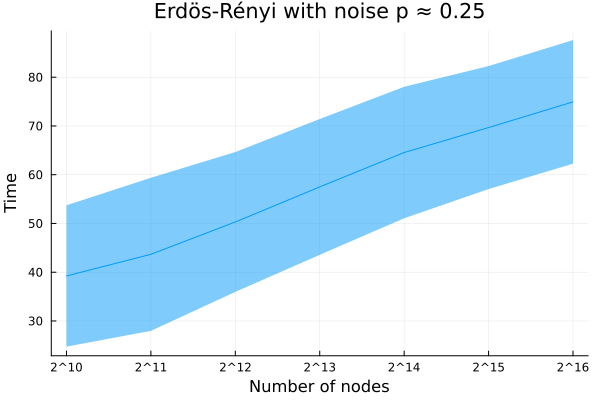}
        \caption{Average convergence time to almost-consensus in an Erdos-Rényi graph \(G_{n,q}\) with \(q = 1.25 \cdot \ln(n)/n \) and noise parameter \(p = 1/4\).}
        \label{fig:experiments:vic-majority:-erdos:d}
    \end{subfigure}
    \caption{Average convergence times to almost-consensus in s and Erdos-Rényi graphs. The average is computed over 1000 different runs, and the random underlying graph is sampled at each run. The shaded areas represent the sample standard deviations.}
    \label{fig:experiments:vic-majority:-erdos}
\end{figure}

\begin{figure}[ht]
    \centering
    \begin{subfigure}{0.4\textwidth}
        \centering
        \includegraphics[scale=.35]{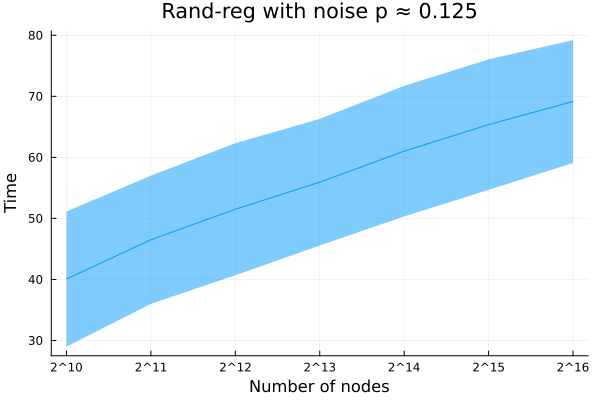}
        \caption{Average convergence time to almost-consensus in a configuration graph with degree \(d = 5\) with noise parameter \(p = 1/8\).}
        \label{fig:experiments:vic-majority:rand-reg:a}
    \end{subfigure}
    \hspace*{1cm}
    \begin{subfigure}{0.4\textwidth}
        \centering
        \includegraphics[scale=.35]{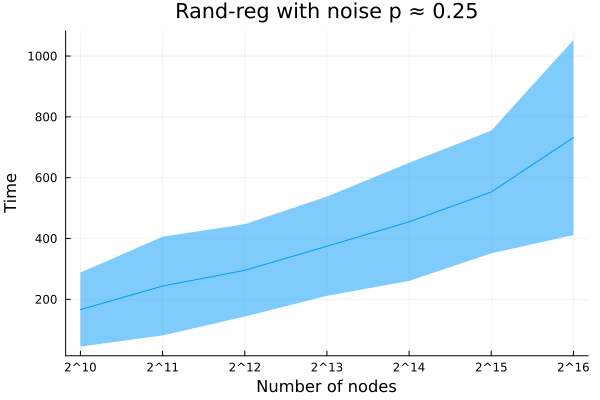}
        \caption{Average convergence time to almost-consensus in a configuration graph with degree \(d = 5\) with noise parameter \(p = 1/4\).}
        \label{fig:experiments:vic-majority:rand-reg:b}
    \end{subfigure}
    \\
    \bigskip
    \begin{subfigure}{0.4\textwidth}
        \centering
        \includegraphics[scale=.35]{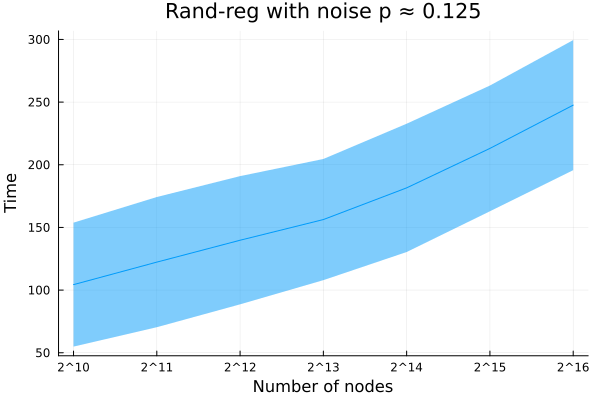}
        \caption{Average convergence time to almost-consensus in a configuration graph with degree \(d = 3\) and noise parameter \(p = 1/8\).}
        \label{fig:experiments:vic-majority:rand-reg:c}
    \end{subfigure}
    \hspace*{1cm}
    \begin{subfigure}{0.4\textwidth}
        \centering
        \includegraphics[scale=.35]{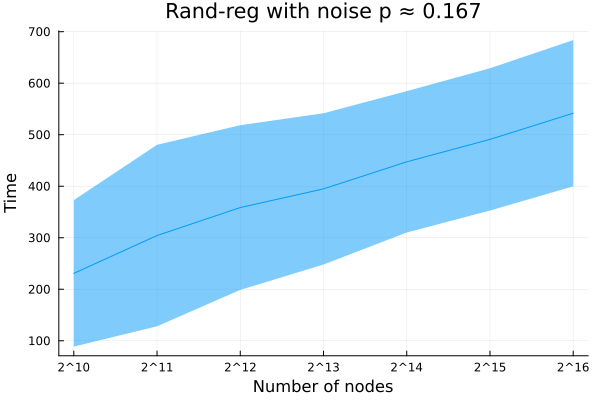}
        \caption{Average convergence time to almost-consensus in a configuration graph with degree \(d = 3\) and noise parameter \(p = 1/6\).}
        \label{fig:experiments:vic-majority:rand-reg:d}
    \end{subfigure}
    \caption{Average convergence times to almost-consensus in random regular graphs with degrees 5 and 3. The average is computed over 1000 different runs, and the random underlying graph is sampled at each run. The shaded areas represent the sample standard deviations.}
    \label{fig:experiments:vic-majority:rand-reg}
\end{figure}

\begin{figure}[ht]
    \centering
    \begin{subfigure}{0.4\textwidth}
        \centering
        \includegraphics[scale=.35]{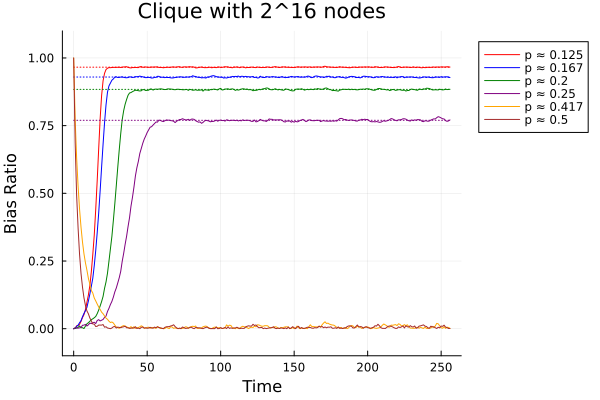}
        \caption{Behavior of the bias in a clique for different noise values.}
        \label{fig:experiments:diff-noise:clique}
    \end{subfigure}
    \hspace*{1cm}
    \begin{subfigure}{0.4\textwidth}
        \centering
        \includegraphics[scale=.35]{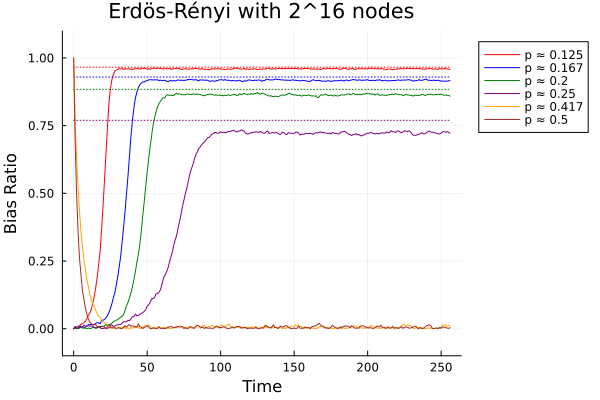}
        \caption{Behavior of the bias in an Erdos-Rényi graph \(G_{n,q}\) with \(q = 1.25 \cdot \ln(n)/n \) for different noise values.}
        \label{fig:experiments:diff-noise:erdos}
    \end{subfigure}
    \\
    \bigskip
    \begin{subfigure}{0.4\textwidth}
        \centering
        \includegraphics[scale=.35]{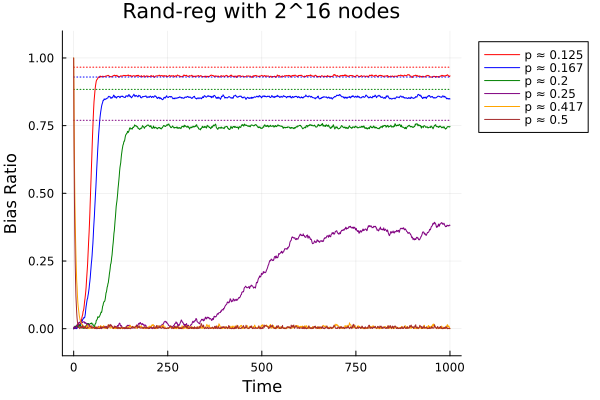}
        \caption{Behavior of the bias in  random regular graph with degree \(d = 5\) for different noise values.}
        \label{fig:experiments:diff-noise:rand-reg:c}
    \end{subfigure}
    \hspace*{1cm}
    \begin{subfigure}{0.4\textwidth}
        \centering
        \includegraphics[scale=.35]{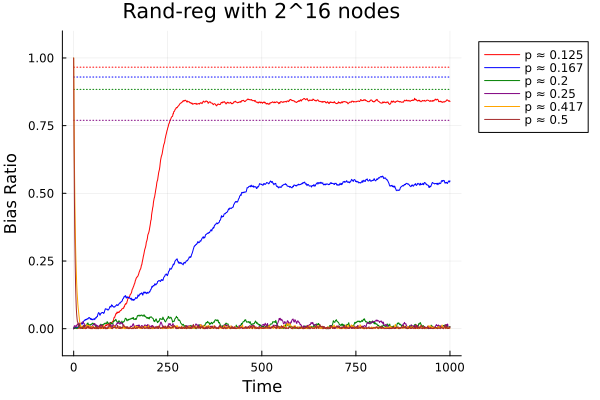}
        \caption{Behavior of the bias in  a random regular graph with degree \(d = 3\) for different noise values.}
        \label{fig:experiments:diff-noise:rand-reg:d}
    \end{subfigure}
    \caption{Behaviors of the bias ratio \(\abs{\frac{\text{bias}}{\text{size}}}\) in different topologies and for different noise values over a single run of the dynamics. The noise parameter \(p\) is chosen from set \(\{1/8,1/6,1/5,1/4,5/12,1/2\}\).
    The dotted colored lines represent the bias ratio's equilibria in cliques, depending on the corresponding noise value.
    All graphs have \(2^{16}\) nodes.
    For the sake of readability, we report here the values of the noise parameters for each colored line: {\color{red} red} = 1/8, 
    {\color{blue} blue} = 1/6, {
    \color{Green} green} = 1/5, 
    {\color{violet} purple} = 1/4, 
    {\color{YellowOrange} gold} = 5/12, 
    {\color{Brown} brown} = 1/2.}
    \label{fig:experiments:diff-noise}
\end{figure}

\clearpage

	\bibliographystyle{abbrv}
	
	\bibliography{biblio}

	\appendix
\section{Tools}\label{app:tools}

We make use of the following general result on (super/sub)-martingales, which can be found in \cite{lehre2014}.

\begin{lemma}\label{lemma:drift_analysis}
	Let $ \{ X_t\}_{t\in \nat} $ be a stochastic process adapted to a filtration $ \{\filtration{t}\}_{t\in \nat} $, over some state space $ S \subseteq \{0\}\cup [\xmin, \xmax] $, where $ \xmin \ge 0 $. Let $  h : [\xmin, \xmax] \to \real^+ $ be a function such that $ 1/h(x) $ is integrable and $ h(x) $ differentiable on $ [\xmin, \xmax] $. Define $ T \coloneqq \min\{t \mid X_t = 0\}$. Then, the followings hold.
	\begin{enumerate}[(i)]
		\item If $ \expect{X_t - X_{t+1} \mid X_t \ge \xmin, \filtration{t} } \ge h(X_t) $ and $ \frac{\diff}{\diff x} h(x) \ge 0$, then 
		\[
		\expect{T \mid X_0} \le \frac{\xmin}{h(\xmin)} + \int_{\xmin}^{X_0} \frac{1}{h(y)} \ \diff y \text{.} 
		\] \label{item:expect-convergence-drift-upbound}
		\item If $ \expect{X_t - X_{t+1} \mid X_t \ge \xmin, \filtration{t} } \le h(X_t) $ and $ \frac{\diff}{\diff x} h(x) \le 0$, then 
		\[
		\expect{T \mid X_0} \ge \frac{\xmin}{h(\xmin)} + \int_{\xmin}^{X_0} \frac{1}{h(y)} \ \diff y \text{.} 
		\] \label{item:expect-convergence-drift-lowbound}
		\item If $ \expect{X_t - X_{t+1} \mid X_t \ge \xmin, \filtration{t} } \ge h(X_t) $ and $ \frac{\diff}{\diff x} h(x) \ge \lambda$ for some $ \lambda > 0 $, then 
		\[
		\pr{T > t \mid X_0} < \exp\left (-\lambda\left (t - \frac{\xmin}{h(\xmin)} - \int_{\xmin}^{X_0} \frac{1}{h(y)} \ \diff y \right )\right ) \text{.} 
		\] \label{item:prob-convergence-drift-upbound}
		\item If $ \expect{X_t - X_{t+1} \mid X_t \ge \xmin, \filtration{t} } \le h(X_t) $ and $ \frac{\diff}{\diff x} h(x) \le - \lambda$ for some $ \lambda > 0 $, then 
		\[
		\pr{T < t \mid X_0} < \frac{e^{\lambda t} - e^\lambda}{e^\lambda - 1}\exp\left (-\lambda\left ( \frac{\xmin}{h(\xmin)} + \int_{\xmin}^{X_0} \frac{1}{h(y)} \ \diff y \right )\right ) \text{.} 
		\] \label{item:prob-convergence-drift-lowbound}
	\end{enumerate}
\end{lemma}

For an overview on the forms of Chernoff bounds see \cite{dubhashi2009}.
\begin{lemma}[Multiplicative forms of Chernoff bounds]\label{lemma:chernoff:multiplicative}
	Let $X_1, X_2, \dots, X_n$ be independent $\{0,1\}$ random variables. Let $X = \sum_{i=1}^n X_i$ and $\mu=\expect{X}$. Then:
	\begin{enumerate}
		\item[(i)] for any $\delta \in (0,1)$ and $\mu \le \mu_+ \le n$, it holds that 
		\begin{align}\label{eq:MCB+}
		\pr{X\ge (1+\delta)\mu_+}\le e^{-\frac{1}{3}\delta^2\mu_+};
		\end{align}
		\item[(ii)] for any $\delta \in (0,1)$ and $0 \le \mu_- \le \mu$, it holds that 
		\begin{align}\label{eq:MCB-}
		\pr{X \le (1-\delta)\mu_-}\le e^{-\frac{1}{2}\delta^2\mu_-}.
		\end{align}
	\end{enumerate}
\end{lemma}

We also make use of the Hoeffding bounds \cite{mitzenmacher2005}.
\begin{lemma}[Hoeffding bounds]\label{lemma:hoeffding}
	Let $ 0 < a < b $ be two constants. Let $X_1, X_2, \dots, X_n$ be independent random variables such that $ \pr{a \le X_i \le b} = 1 $ for all $ i \in [n] $. Let $X = \sum_{i=1}^n X_i$ and $ \expect{X} = \mu$. Then:
	\begin{enumerate}
		\item[(i)] for any $t > 0$ and $\mu \le \mu_+ $, it holds that 
		\begin{align}\label{eq:HB+}
		\pr{X \ge \mu_+ + t} \le \exp\left (- \frac{2t^2 }{n(b-a)^2}\right );
		\end{align}
		\item[(ii)] for any $t > 0$ and $0 \le \mu_- \le \mu$, it holds that 
		\begin{align}\label{eq:HB-}
		\pr{X \le \mu_- - t} \le \exp\left (- \frac{2t^2 }{n(b-a)^2}\right ).
		\end{align}
	\end{enumerate}
\end{lemma}

We make use of the following result which makes explicit the convergence ``speed'' in the central-limit theorem.
\begin{lemma}[Berry-Esseen]\label{lemma:berry-eseen}
	Let $X_1, \dots, X_n$ be $n$ i.i.d.\ (either discrete or continuous) random variables with zero mean, variance $\sigma^2>0$, and finite third moment. Let $Z$ the standard normal random variable, with zero mean and variance equal to 1. Let $F_n(x)$ be the cumulative function of $ \frac{S_n}{\sigma\sqrt{n}}$, where $S_n = \sum_{i=1}^n X_i$, and $\Phi(x)$ that of $Z$. Then, there exists a positive constant $C>0$ such that
	\[
	\sup_{x\in \real} \abs{F_n(x) - \Phi(x)} \le \frac{C}{\sqrt{n}}
	\]
	for all $n\ge 1$.
\end{lemma}

Finally, we use some anti-concentration inequalities know as reverse Chernoff bounds. The proof can be found in the appendix of \cite{klein1999O}.

\begin{lemma}[Reverse Chernoff bounds]\label{lemma:chernoff-reverse}
	Let $X_1, X_2, \dots, X_n$ be i.i.d.\  $\{0,1\}$ random variables. Let $X = \sum_{i=1}^n X_i$ and $\mu=\expect{X}$, with $\mu \le n/2$. Furthermore, let $\delta \in (0, 1/2]$ be a constant. If $\delta^2 \mu \ge 3$, then:
	\begin{enumerate}
		\item[(i)] for any $\mu \le \mu_+ \le n$, it holds that 
		\begin{align}\label{eq:RCB+}
		\pr{X\ge (1+\delta)\mu_+}\ge e^{-9 \delta^2 \mu_+};
		\end{align}
		\item[(ii)] for any $0 \le \mu_- \le \mu$, it holds that 
		\begin{align}\label{eq:RCB-}
		\pr{X \le (1-\delta)\mu_-}\ge e^{-9\delta^2\mu_-}.
		\end{align}
	\end{enumerate}
\end{lemma}
 \end{document}